\documentclass{eptcs}
\providecommand{\event}{ICE 2013} 
\usepackage{breakurl}             

\usepackage{graphicx,color}
\usepackage{amsmath,amsthm,amscd}
\usepackage{latexsym,amssymb}
\usepackage{enumitem}
\usepackage{tikz}
\usepackage{wrapfig,caption,subcaption}
\usetikzlibrary{matrix,arrows,decorations.pathmorphing}

\newtheorem{theorem}{Theorem}[section]
\newtheorem{proposition}[theorem]{Proposition}
\newtheorem{lemma}[theorem]{Lemma}
\newtheorem{corollary}[theorem]{Corollary}

\theoremstyle{definition}
\newtheorem{definition}[theorem]{Definition}

\newtheorem{example}[theorem]{Example}
\newtheorem{note}[theorem]{Remark}
\newtheorem*{axioms}{Axioms}
\newtheorem*{syntax}{Syntax}
\newtheorem*{semantics}{Semantics}

\newcommand{\defn}[1]{Definition~\ref{defn:#1}}
\newcommand{\fig}[2][]{Figure~\ref{fig:#2}\ensuremath{#1}}

\newcommand{\eq}[1]{(\ref{eq:#1})}

\newcommand{\ex}[1]{Example~\ref{ex:#1}}
\newcommand{\secn}[1]{Section~\ref{sec:#1}}

\newcommand{\lem}[1]{Lemma~\ref{lem:#1}}

\newcommand{\axs}[1]{Axioms~\ref{ax:#1}}
\newcommand{\axss}[2]{Axioms~\ref{ax:#1}, \ref{ax:#2}}
\newcommand{\ax}[1]{Axiom~\ref{ax:#1}}
\newcommand{\prop}[1]{Proposition~\ref{prop:#1}}

\newcommand{\mycaption}[1]{\caption{#1.}}

\newcommand{\mdash}[1][]{---#1}

\newcommand{\ie}{i.e.\ }
\newcommand{\eg}{e.g.\ }
\newcommand{\cf}{cf.\ }

\newcommand{\bydef}[1]{\ensuremath{\stackrel{\Delta}{#1}}}

\newcommand{\setdef}[2]{\ensuremath{\{{#1}\,|\,{#2}\}}}
\newcommand{\Setdef}[2]{\ensuremath{\Big\{{#1}\,\Big|\,{#2}\Big\}}}
\newcommand{\fire}[1]{\ensuremath{\dot{#1}}}
\newcommand{\act}[1]{\ensuremath{#1}}
\newcommand{\inhib}[1]{\ensuremath{\overline{#1}}}
\newcommand{\non}[1]{\ensuremath{\overline{#1}}}
\newcommand{\firesup}[1]{\ensuremath{\mathbf{fire}{\left(#1\right)}}}
\newcommand{\actsup}[1]{\ensuremath{\mathbf{act}{\left(#1\right)}}}
\newcommand{\negsup}[1]{\ensuremath{\mathbf{neg}{\left(#1\right)}}}
\newcommand{\offer}[1][]{\ensuremath{\!\uparrow\!{#1}}}
\newcommand{\noffer}[1][]{\ensuremath{\!\not\,\uparrow\!{#1}}}
\newcommand{\goesto}[1][]{\ensuremath{\stackrel{#1}{\rightarrow}^{}}}
\newcommand{\longgoesto}[1][]{\ensuremath{\stackrel{#1}{\longrightarrow}}}
\newcommand{\symmdiff}[2]{\ensuremath{{#1}\triangle{#2}}}
\newcommand{\glue}[1][]{\ensuremath{{\cal G}^{#1}}}

\newcommand{\less}{\prec}

\newcommand{\true}{\ensuremath{\mathtt{tt}}}
\newcommand{\false}{\ensuremath{\mathtt{ff}}}

\newcommand{\intsem}[1]{\ensuremath{\|{#1}\|}}
\newcommand{\aisem}[1]{\ensuremath{|{#1}|}}

\newcommand{\ai}{\ensuremath{\mathcal{A\hspace{-0.6ex}I\!}}}
\newcommand{\ct}{\ensuremath{\mathcal{T\!}}}
\newcommand{\cru}{\ensuremath{\mathcal{CR\!}}}
\newcommand{\ac}{\ensuremath{\mathcal{AC}\!}}

\newcommand{\cA}{\ensuremath{\mathcal{A}}}
\newcommand{\sB}{\ensuremath{\mathbb{B}}}
\newcommand{\bB}{\ensuremath{\mathbf{B}}}
\newcommand{\cB}{\ensuremath{\mathcal{B}}}

\newcommand{\derrule}[3][1]{%
  \ensuremath{%
    \begin{array}{*{#1}{@{\hspace{2mm}}c@{\hspace{2mm}}}}
      #2\\
      \hline
      \multicolumn{#1}{c}{#3}
    \end{array}%
  }%
}

\newlength{\algowidth}
\title{Extended Connectors: Structuring Glue Operators in BIP}

\author{Eduard Baranov and Simon Bliudze
\institute{\'Ecole Polytechnique F\'ed\'erale de Lausanne\\
Rigorous System Design Laboratory\\
INJ Building, Station 14, 1015 Lausanne, Switzerland}
\email{\{firstname.lastname\}@epfl.ch}
}
\def\titlerunning{Extended Connectors in BIP}
\def\authorrunning{E.~Baranov, S.~Bliudze}

\begin{document}

\maketitle

\begin{abstract}
Based on a variation of the BIP operational semantics using the offer
predicate introduced in our previous work, we extend the algebras used to
model glue operators in BIP to encompass priorities. This extension uses
the Algebra of Causal Interaction Trees, $\ct(P)$, as a pivot: existing
transformations automatically provide the extensions for the Algebra of
Connectors.  We then extend the axiomatisation of $\ct(P)$, since the
equivalence induced by the new operational semantics is weaker than that
induced by the interaction semantics.  This extension leads to canonical
normal forms for all structures and to a simplification of the algorithm
for the synthesis of connectors from Boolean coordination constraints.
\end{abstract}


\section{Introduction}
\label{sec:intro}

Component-based design is based on the separation between coordination and
computation.  Systems are built from units processing sequential code
insulated from concurrent execution issues.  The isolation of coordination
mechanisms allows a global treatment and analysis.

Fundamentally, each component-based design framework consists of a
behaviour type \cB{} \cite{bliudze12-glue}, defining the underlying
semantic domain and the key properties such as the relevant equivalence
relations, and a set \glue{} of glue operators of the form $f: 2^\cB
\rightarrow \cB$.  As argued in \cite{framework05}, an important property
for glue operators is the possibility to be {\em flattened}: given a
behaviour $g(f(B_1,\dots,B_k), B_{k+1},\dots,B_n)$ obtained by hierarchical
composition with two glue operators, there must be an equivalent\footnote{
  The notion of equivalence, in this context, is given by
  the behaviour type \cB{} \cite{bliudze12-glue}.
} behaviour $h(B_1,\dots,B_n)$ obtained by applying a single glue
operator to the {\em same} atomic components.  In other words, \glue{} must
be closed under composition.  Flattening enables model transformations, \eg
for optimising code generation or component placement on multicore
platforms \cite{quilbeuf10-distr,jaber09-s2s}. 

BIP is a component framework for constructing systems by superposing three
layers: Behaviour, Interaction and Priorities.  In the classical BIP
semantics \cite{BliSif07-acp-emsoft}, behaviour is modelled by
\emph{Labelled Transition Systems} (LTS), \ie triples $B=(Q,P,\goesto)$,
where $Q$ is a set of {\em states}, $P$ is a set of {\em ports}, and
$\goesto\, \subseteq Q\times 2^P\times Q$ is a set of {\em transitions},
each labelled by an interaction (a subset of ports).  Glue operators are
defined using interaction and priority models.

For a set of behaviours $\setdef{B_i=(Q_i, P_i, \goesto)}{i\in [1,n]}$, an
\emph{interaction model} is a set of \emph{interactions} $\gamma \subseteq 2^P$,
where $P = \bigcup_{i=1}^n P_i$ (all $P_i$ are assumed to be pairwise disjoint).
The behaviour $\gamma(B_1,\dots,B_n)$ is defined by the behaviour $(Q, P,
\goesto_\gamma)$, with $Q = \prod_{i=1}^n Q_i$ and the minimal transition relation
$\goesto_\gamma$ satisfying the rule (we use set notation to group premises of the same type)
\begin{equation}
  \label{eq:transsem}
  \derrule[3]{a \in \gamma &
    \Setdef{q_i \longgoesto[a \cap P_i] q'_i}
           {a \cap P_i \not= \emptyset,\ i\in [1,n]} &
    \Setdef{q_i = q'_i}
           {a \cap P_i = \emptyset,\ i\in [1,n]}
  }{
      q_1\dots q_n \longgoesto[a]_\gamma q'_1\dots q'_n
  }\,.
\end{equation}

For a behaviour $B = (Q,P,\goesto)$, a {\em priority model} is a strict
partial order $\less$ on $2^P$.  When $a \less a'$, we say that the interaction $a'$ has \emph{higher priority} than $a$.  We put $B_\less \bydef{=}
(Q, P, \goesto_\less)$, with the minimal transition relation $\goesto_\less$
satisfying the rule
\begin{equation}
  \label{eq:prisem}
  \derrule[2]{
    q \longgoesto[a] q' &
     \Setdef{q \not\longgoesto[a']}{a \less a'}
  }{
    q \longgoesto[a]_\less q'
  }\,.
\end{equation}

Each $n$-ary glue operator in BIP is obtained as the composition of an interaction model
(an $n$-ary operator), composing several behaviours into a single one, and
a unary priority model.\footnote{
  Notice that both interaction and priority models can be trivial: a trivial interaction model over the set of ports $P$ is the set of \emph{singleton interactions} $\{\{p\}\,|\, p \in P\}$; a trivial priority model is empty with none of the interactions having higher priority than any other.
}  In general, when combined hierarchically such glue
operators cannot be flattened. Indeed, consider the following example.

\begin{figure}
  \centering
  \begin{subfigure}[b]{0.40\textwidth}
    \centering
    \input{nonflat-atoms.pspdftex}
    \caption{Atomic components $B_1$, $B_2$ and $B_3$}
    \label{fig:nonflat:atoms}
  \end{subfigure}%
  \hfill 
  \begin{subfigure}[b]{0.24\textwidth}
    \centering
    \input{nonflat-composed.pspdftex}
    \caption{Composed system}
    \label{fig:nonflat:composed}
  \end{subfigure}
  \hfill 
  \begin{subfigure}[b]{0.24\textwidth}
    \centering
    \input{nonflat-flat.pspdftex}
    \caption{Flat system}
    \label{fig:nonflat:flat}
  \end{subfigure}
  \mycaption{BIP component that cannot be flattened (\ex{nonflat})}
  \label{fig:nonflat}
\end{figure}

\begin{example}
  \label{ex:nonflat}
  Let $B_1$, $B_2$ and $B_3$ be the three atomic behaviours shown in
  \fig{nonflat:atoms} and consider the composed behaviour
  $g(f(B_1,B_2),B_3)$ (\fig{nonflat:composed}), with the glue operator $f$
  defined by the interaction model $\{p,q,r,s\}$ (omitted in
  \fig{nonflat:composed}) and priority model $\{p \less r\}$; $g$ defined
  by the interaction model $\{p,q,s,rt\}$ without any additional priority
  model.  One can prove that it is not possible to represent
  this behaviour as a flat one (\fig{nonflat:flat}).  Indeed, it is not
  sufficient to replace the priority $p \less r$ by $p \less rt$: in the
  global state $(1,3,6)$ of the composed behaviour in
  \fig{nonflat:composed}, interaction $p$ is inhibited by the priority $p
  \less r$; in the same state of the composed behaviour in
  \fig{nonflat:flat}, $p$ would not be inhibited by $p \less rt$, since
  interaction $rt$ is not enabled.

  Furthermore, although this goes beyond the scope of this paper, one can
  prove that there is no flat glue operator $h$ in the classical BIP semantics
  given by \eq{transsem} and \eq{prisem}, such that $g(f(B_1,B_2),B_3)$ be equivalent to
  $h(B_1,B_2,B_3,B_4)$ with any additional helper behaviour $B_4$.
\end{example}

The impossibility of flattening in the above example is due to the fact that the information used by
the priority model refers only to interactions authorised by the underlying
interaction model.  All the information about interactions enabled in
atomic components is lost after the application of $f$.  For instance, one can consider that, in
\ex{nonflat}, transitions $p$ and $r$ model respectively taking and
liberating a semaphore.  Thus $p$ should be disabled whenever $r$ is
possible, independently of whether $r$ can actually be taken on not (\eg
when $r$ is blocked waiting for a synchronisation, as in
\fig{nonflat:composed}).

In \cite{BliSif11-constraints-sc}, a variation of the BIP operational
semantics was introduced. In this variation, a behaviour is defined as an
LTS with an additional {\em offer} predicate, \ie a quadruple
$B=(Q,P,\goesto[],\offer[])$, such that $\offer \subseteq Q \times P$ and,
for any $q\in Q$ and $p \in P$, holds the implication $(\exists a \subseteq
P: p\in a \land q\goesto[a]) \implies q\offer[p]$.  The converse
implication is not required. In particular, when a transition labelled $p$
in a sub-component of a composed behaviour is blocked waiting for a
synchronisation, $p$ is still considered as offered.\footnote{
  As a side effect, the offer predicate can be used to distinguish between
  atomic and composite behaviours: a behaviour is atomic iff $(\exists a
  \subseteq P: p\in a \land q\goesto[a]) \Longleftrightarrow q\offer[p]$
  \cite{BliSif11-constraints-sc}.
} In \cite{BliSif11-constraints-sc}, we have established the equivalence between, on one hand, 
glue operators defined by sets of SOS rules having positive premises in terms of the transition relations $\goesto[]$ and offering predicates $\offer[]$ and negative premises in terms of the offering predicates only, and, on the other hand, Boolean formula with the so-called {\em firing} and
{\em activation} variables. We have also studied the expressiveness of such glue operators and compared it with classical BIP.

Using the offer predicate instead of the transition relation in the
negative premises of \eq{prisem}, ensures that the resulting set of glue operators is closed under composition.

In this paper, we show how the algebras representing interaction models
can be naturally generalised to also define priorities, based on the use of activation and firing variables.

Several algebraic structures are used to define and manipulate interaction
models in BIP \cite{BliSif07-acp-emsoft, BliSif10-causal-fmsd}.
  
{\bf The Algebra of Interactions}, $\ai(P)$, is isomorphic to $2^{2^P}$.
It provides a simple algebraic representation of interaction models
simplifying the definition of the semantics of other algebras.
  
{\bf The Algebra of Connectors}, $\ac(P)$, defines the connectors in the
form used in the BIP language, well adapted for graphical representation
and for the specification of data transfers.
  
{\bf The Algebra of Causal Interaction Trees}, $\ct(P)$, defines an
alternative semantic domain for connectors with the explicit causality
relation between ports.  Coherence results for the
$\ai(P)$ and $\ct(P)$ semantics of connectors have been provided in
\cite{BliSif10-causal-fmsd}.
  
{\bf Systems of Causal Rules}, $\cru(P)$, derived from causal interaction
trees define a Boolean representation of connectors, suitable for symbolic
manipulation and for specification of state safety properties.

In \cite{BliSif10-causal-fmsd}, the four transformations were provided
between $\ac(P)$ and $\ct(P)$, and between $\ct(P)$ and $\cru(P)$.  In
particular, this allows to synthesise connectors from $\sB[P]$ Boolean
formul\ae.

In this paper, we study the extension of the above algebras to represent both interaction and priority models.  Equivalence induced by the new operational semantics is weaker
than that induced by the interaction semantics.  We extend accordingly the
axioms of $\ct(\cdot)$ and provide corresponding normal forms for terms of
the considered algebras.  Finally, we show that, in this context, the
connector synthesis algorithm in \cite{BliSif10-causal-fmsd} can be
simplified by considering only the causal rules with firing variables in
the effect.

The rest of the paper is structured as follows.  We start, in \secn{related}, by a short discussion of some related work.  In \secn{representations}, we briefly recall the syntax and semantics of
all the considered algebras.  \secn{semantic} presents the new semantic
model for BIP based on the offer predicate.  Main contributions of the paper, namely the
extensions of the algebras encompassing the activation and negative ports
are presented in \secn{extension}.  We illustrate the extended algebras with
a connector synthesis example presented in \secn{example}.  Finally, \secn{conclusion} concludes the
paper.


\section{Related work} 
\label{sec:related}

The results in this paper build on our previous work cited above.  However,
the following related work should also be mentioned.  The approach we use
for the Boolean encoding of glue constraints is close to that used for
computing flows in Reo connectors in \cite{DeconstructingReo}, where it is
further extended to data flows.  

Several methodologies for synthesis of component coordination have been
proposed in the literature, \eg connector synthesis in
\cite{Arbab05ReoSynth, arbab08-synth, inverardi01}.  Both approaches are
very different from ours.  In \cite{Arbab05ReoSynth}, Reo circuits are
generated from constraint automata.  This approach is limited, in the first
place, by the complexity of building the automaton specification of
interactions.  An attempt to overcome this limitation is made in
\cite{arbab08-synth} by generating constraint automata from UML sequence
diagrams.  In \cite{inverardi01}, connectors are synthesised in order to
ensure deadlock freedom of systems that follow a very specific
architectural style imposing both the interconnection topology and
communication primitives (notification and request messages).

Recently a comparative study \cite{bruni12-tiles-wire-BIP} of three
connector frameworks\mdash tile model \cite{montanari06}, wire calculus
\cite{sobocinski09-wire} and BIP\mdash has been performed.  From the
operational semantics perspective, this comparison only accounts for
operators with positive premises.  In particular, priority in BIP is not
considered.  It would be interesting to see whether using ``local'' offer
predicate instead of ``global'' priorities of the classical BIP could help
generalising this work.


\section{Representations of the interaction model} 
\label{sec:representations}

In this section, we briefly recall the syntax and semantics of the algebras
used to represent BIP interaction models.  The semantics of the Algebra of
Interactions is given in terms of sets of interactions by a function
$\intsem{\cdot}: \ai(P) \rightarrow 2^{2^P}$.  Two terms $x,y \in \ai(P)$
are {\em equivalent} $x \simeq y$ iff $\intsem{x} = \intsem{y}$.  For any
other algebra, $\cA(P)$, among those mentioned in the introduction, we
define its semantics by the function $\aisem{\cdot}: \cA(P) \rightarrow
\ai(P)$.  A function $\intsem{\cdot}: \cA(P) \rightarrow 2^{2^P}$ is
obtained by composing $\aisem{\cdot}: \cA(P) \rightarrow \ai(P)$ and
$\intsem{\cdot}: \ai(P) \rightarrow 2^{2^P}$.  The axiomatisation of
$\ai(P)$ given in \cite{BliSif07-acp-emsoft} is sound and complete with
respect to $\simeq$.  Hence, for other algebras, the equivalences induced
by $\intsem{\cdot}$ and $\aisem{\cdot}$ coincide.

Below, we assume that a set of ports $P$ is given, such that $0,1\not\in
P$.


\subsection{Algebra of Interactions}
\label{sec:ai}

\begin{syntax}
The syntax of the {\em Algebra of Interactions}, $\ai(P)$, is defined by
the following grammar
\begin{equation} \label{eq:apsynt}
  \begin{array}{rc*{5}{l@{\ |\ }}l}
    x & ::= & 0 & 1 & p \in P & x\cdot x & x + x & (x)\,,\\
  \end{array}
\end{equation}
where `$+$' and `$\cdot$' are binary operators, respectively called {\em
  union} and {\em synchronisation}.  Synchronisation binds stronger than
union.
\end{syntax}

\begin{semantics}
The semantics of $\ai(P)$ is given by the function $\|\cdot\| : \ai(P)
\rightarrow 2^{2^P}$, defined by
\begin{equation} \label{eq:apsem}
  \renewcommand{\arraystretch}{1.5}
  \begin{array}{lcl}
    \lefteqn{\|0\|\ =\ \emptyset,\quad \|1\|\ =\ \{\emptyset\},\quad
      \|p\|\ =\ \Big\{\{p\}\Big\},}\\
    \|x_1 + x_2\| &=& \|x_1\| \cup \|x_2\|,\\
    \|x_1 \cdot x_2\| &=& 
    \Big\{a_1 \cup a_2\,\Big|\, a_1\in \|x_1\|, a_2\in \|x_2\| \Big\},\\
    \|(x)\| & = & \|x\|,
  \end{array}
\end{equation}
for $p\in P$, $x,x_1,x_2\in\ai(P)$.  Terms of $\ai(P)$ represent sets of
interactions between the ports $P$.
\end{semantics}

Sound and complete axiomatisation of $\ai(P)$ with respect to the semantic
equivalence is provided in \cite{BliSif07-acp-emsoft}.  In a nutshell,
$(\ai(P), +, \cdot, 0, 1)$ is a commutative semi-ring idempotent in both
$+$ and $\cdot$.


\subsection{Algebra of Connectors}
\label{sec:ac}

\begin{syntax}
The syntax of the {\em Algebra of Connectors}, $\ac(P)$, is defined by the
following grammar
\begin{equation} \label{eq:acsynt}
  \renewcommand{\arraystretch}{1.5}
  \begin{array}{lcll}
    s & ::= & [0]\ |\ [1]\ |\ [p]\ |\ [x] \:\:\:\:\:\:\:\:\:(synchrons)\\
    t & ::= & [0]'\ |\ [1]'\ |\ [p]'\ |\ [x]'  \:\:\:\:\:\:   (triggers)\\
    x & ::= & s\ |\ t\ |\ x\cdot x\ |\ x + x\ |\ (x)\,, 
  \end{array}
\end{equation}
for $p\in P$, and where `$+$' is binary operator called {\em union},
`$\cdot$' is a binary operator called {\em fusion}, and brackets
`$[\cdot]$' and `$[\cdot]'$' are unary {\em typing} operators.  Fusion
binds stronger than union.

Fusion is a generalisation of the synchronisation in $\ai(P)$.  Typing is
used to form typed connectors: `$[\cdot]'$' defines {\em triggers} (can
initiate an interaction), and `$[\cdot]$' defines {\em synchrons} (need
synchronisation with other ports in order to interact).
\end{syntax}

\begin{semantics}
The semantics of $\ac(P)$ is given by the function $|\cdot| : \ac(P)
\rightarrow \ai(P)$:
\begin{align}
  \label{eq:acap}
  |[p]| & = p\,,
  & |x_1 + x_2| & = |x_1| + |x_2|\,,
  & \Big|\prod_{i=1}^n [x_i]\,\Big| & = \prod_{i=1}^n |x_k|\,,
  \\
  \label{eq:acaptrig} 
  && \hspace{-2cm}\Big|\prod_{i=1}^n [x_i]' \prod_{j=1}^m [y_j]\,\Big| = 
  \sum_{i=1}^n |x_i| & \lefteqn{\left(
  \prod_{k\not=i}\Big(1 + |x_k|\Big)\ \prod_{j=1}^m \Big(1 + |y_j|\Big)\right)}\,,
\end{align}
for $x,x_1,\dots,x_n,y_1,\dots,y_m \in \ac(P)$ and $p\in P\cup \{0,1\}$.
\end{semantics}

Sound and complete axiomatisation of $\ac(P)$ with respect to the semantic
equivalence is provided in \cite{BliSif08-acp-tc}.  We omit it here, since
we will not need it in the rest of this paper.

\fig{connectors} shows four basic examples of the graphical representation
of connectors.  Triggers are denoted by triangles, whereas synchrons are
denoted by bullets.  The interaction semantics of the four connectors is
given in the subfigure captions.

\begin{figure}
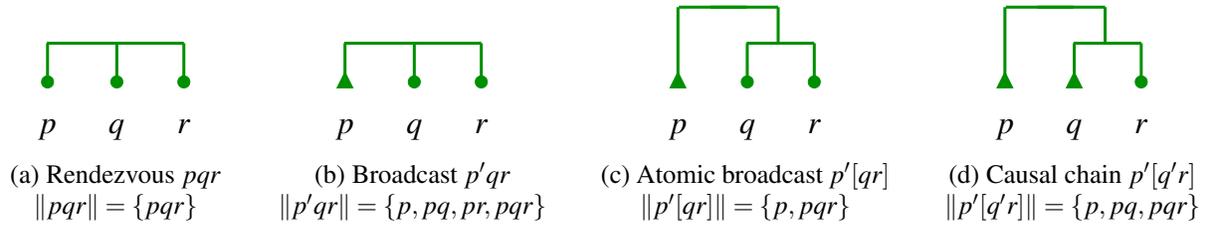

  \centering
  \begin{subfigure}[b]{0.18\textwidth}
    \centering
    \input{rdv.pspdftex}
    \caption{\centering Rendezvous~$pqr$
      $\intsem{pqr} = \{pqr\}$}
    \label{fig:connectors:rdv}
  \end{subfigure}%
  \hfill 
  \begin{subfigure}[b]{0.24\textwidth}
    \centering
    \input{broadcast.pspdftex}
    \caption{\centering Broadcast~$p'qr$
      $\intsem{p'qr} = \{p, pq, pr, pqr\}$}
    \label{fig:connectors:bdc}
  \end{subfigure}%
  \hfill 
  \begin{subfigure}[b]{0.24\textwidth}
    \centering
    \input{atomic-bdc.pspdftex}
    \caption{\centering Atomic broadcast~$p'[qr]$
      $\intsem{p'[qr]} = \{p, pqr\}$}
    \label{fig:connectors:atomic}
  \end{subfigure}%
  \hfill 
  \begin{subfigure}[b]{0.23\textwidth}
    \centering
    \input{causal-chain.pspdftex}
    \caption{\centering Causal chain~$p'[q'r]$
      $\intsem{p'[q'r]} = \{p, pq, pqr\}$}
    \label{fig:connectors:causal}
  \end{subfigure}%
  \caption{Basic connector examples}
  \label{fig:connectors}
\end{figure}


\subsection{Algebra of Causal Interaction Trees}
\label{sec:trees}

\begin{syntax}
The syntax of the \emph{Algebra of Causal Interaction Trees}, $\ct(P)$, is
given by
\begin{equation} \label{eq:ctsyn}
  t ::= a \,|\, a \rightarrow t \,|\, t\oplus t\,,
\end{equation}
where $a \in 2^P \cup \{0,1\}$ is an interaction, and `$\rightarrow$' and
`$\oplus$' are respectively the {\em causality} and the {\em parallel
  composition} operators.  Causality binds stronger than parallel
composition.  Notice that a causal interaction tree can have several roots.

The causality operator is right- (but not left-) associative, thus for interactions
$a_1,\dots,a_n$, we can abbreviate $a_1 \rightarrow (a_2 \rightarrow (\dots
\rightarrow a_n)\dots))$ to $a_1 \rightarrow a_2 \rightarrow \dots
\rightarrow a_n$.  We call this construction a {\em causal chain}.
\end{syntax}

\begin{semantics}
The semantics of $\ct(P)$ is given by the function $|\cdot|: \ct(P)
\rightarrow \ai(P)$
\begin{align}
  \label{eq:ctsem}
  |a| & = a\,,
  & |a \rightarrow t| & = a\Big(1 + |t|\Big)\,,
  & |t_1 \oplus t_2| & = |t_1| + |t_2| + |t_1|\,|t_2|\,,
\end{align}
where $a$ is an interaction and $t, t_1, t_2 \in \ct(P)$.
\end{semantics}

A sound (although not complete) axiomatisation of $\ct(P)$ is provided in
\cite{BliSif10-causal-fmsd}.  Rather than reproduce it here, we indicate
the differences after the extension provided in \secn{refinement}.


\subsection{Systems of Causal Rules}
\label{sec:rules}

Below, for any set $X$ of propositional variables, we denote by $\sB[X]$ the corresponding Boolean algebra generated by $X$.  For presentation clarity, we will often omit the conjunction operator and write $a \lor bc$ instead of $a \lor (b \land c)$.

\begin{definition}
  A {\em causal rule} is a $\sB[P]$ formula $E \Rightarrow C$, where $E$
  (the \emph{effect}) is either a constant, $\true$, or a port variable $p
  \in P$, and $C$ (the \emph{cause}) is either a constant, $\true$ or
  $\false$, or a positive $\sB[P\setminus\{p\}]$ formula in disjunctive normal form.
\end{definition}

\begin{note} \label{rem:absorption}
  Notice that $a_1 \lor a_1\,a_2 = a_1$, and therefore causal rules can be
  simplified by replacing $p \Rightarrow a_1 \lor a_1\,a_2$ with $p
  \Rightarrow a_1)$.  We assume that all the causal rules are simplified by
  this absorption rule.
\end{note}

\begin{definition}
  A \emph{system of causal rules} is a set $R = \{p \Rightarrow x_p\}_{p\in
    P^t}$, where $P^t \bydef{=} P \cup \{\true\}$, having precisely one causal rule for each port variable $p \in P^t$.  An interaction $a \in
  2^P$ satisfies the system $R$ (denoted $a \models R$), iff the
  characteristic valuation of $a$ on $P$ satisfies the formula
  $\bigwedge_{p\in P^t} (p \Rightarrow x_p)$.  We denote by
  $|R|\ \bydef{=}\ \sum_{a \models R} a$ the union (in terms of the Algebra
  of Interactions) of the interactions satisfying $R$.  Thus we have
  $|\cdot| : \cru(P) \rightarrow \ai(P)$, where $\cru(P)$ is the set of all
  systems of causal rules over the set of port variables $P$.
\end{definition}



\subsection{Transformations between different representations}
\label{sec:transformations}

Transformations $\ac(P) \overset{\tau}{\underset{\sigma}{\rightleftarrows}}
\ct(P)$ and $\ct(P) \overset{R}{\rightleftarrows} \cru(P)$ were defined in
\cite{BliSif10-causal-fmsd} and shown to respect $\simeq$.  Below, we will
only need the transformations $\sigma : \ct(P) \rightarrow \ac(P)$ and $R:
\ct(P) \rightarrow \cru(P)$.  The former is defined recursively by putting
\begin{align}
  \label{eq:treecon}
  \sigma(a) & = [a]\,,
  & \sigma(a \rightarrow t) & = [a]'\,[\sigma(t)]\,,
  & \sigma(t_1 \oplus t_2)  & = [\sigma(t_1)]'\,[\sigma(t_2)]'\,.
\end{align}
We define $R: \ct(P) \rightarrow \cru(P)$ by putting
\begin{equation} 
  \label{eq:trees2rules}
  R(t)\ =\ \{p \Rightarrow c_p(t)\}_{p\in P\cup\{\true\}}\,,
\end{equation}
where the function $c_p : \ct(P) \rightarrow \sB[P]$ is defined recursively
as follows.  For $a\in 2^P$ (with $p\not\in a$) and $t,t_1,t_2 \in \ct(P)$,
we put
\begin{align*}
  c_p(0) & = \false\,, 
  & c_{\true}(0) & = \false\,,\\
  c_p(p \rightarrow t) & = \true\,,
  & c_{\true}(1 \rightarrow t) & = \true\,,\\
  c_p(pa \rightarrow t) & = a\,,
  & c_{\true}(a \rightarrow t) & = a\,,\\
  c_p(a \rightarrow t) & = a \land c_p(t)\,,\\
  c_p(t_1 \oplus t_2) & = c_p(t_1) \lor c_p(t_2)\,, 
  & c_{\true}(t_1 \oplus t_2) & = c_{\true}(t_1) \lor c_{\true}(t_2)\,.
\end{align*}

Observe that this transformation associates to each port $p \in P$ a causal
rule $p \Rightarrow C$, where $C$ is the disjunction of all prefixes leading
from roots of $t$ to some node containing $p$, including the ports of this
node other than $p$.


\section{Modification of the semantic model}
\label{sec:semantic}

We now present the variation of the BIP operational semantics based on the
offer predicate \cite{BliSif11-constraints-sc}.

\begin{definition}
  \label{defn:lts}
  A \emph{labelled transition system} (LTS) is a triple $(Q,P,\goesto)$,
  where $Q$ is a set of \emph{states}, $P$ is a set of \emph{ports}, and
  $\goesto\, \subseteq Q\times 2^P \times Q$ is a set of
  \emph{transitions}, each labelled by a non-empty set of ports.  For $q,q'
  \in Q$ and $a \in 2^P$, we write $q \goesto[a] q'$ iff $(q,a,q') \in\,
  \goesto$.  A label $a \in 2^P$ is \emph{active} in a state $q \in Q$
  (denoted $q \goesto[a]$), iff there exists $q' \in Q$ such that $q
  \goesto[a] q'$.  We abbreviate $q \not\goesto[a] \bydef{=} \lnot (q
  \goesto[a])$.  
\end{definition}

Below, it is assumed that, for all $q \in Q$, $q \goesto[\emptyset]
q$.  All results of the paper can be reformulated without this
assumption, but making it simplifies the presentation.  We write $pq$
for the set of ports $\{ p, q\}$.

\begin{definition}
  \label{defn:behaviour}
  A \emph{behaviour} is a pair $B=(S,\offer)$ consisting of an LTS
  $S=(Q,P,\goesto)$ and an \emph{offer} predicate $\offer$ on $Q \times P$
  such that $q \offer[p]$ holds (a port $p \in P$ is \emph{offered} in a
  state $q \in Q$) whenever there is a transition from $q$ containing $p$,
  that is $(\exists a \in 2^P: p \in a \land q \goesto[a]) \Rightarrow q
  \offer[p]$.  We write $B = (Q,P,\goesto,\offer)$ for $B =
  ((Q,P,\goesto),\offer)$.

  The offer predicate extends to sets of ports: for $a \in 2^P$, $q
  \offer[a] \bydef{=} \bigwedge_{p \in a} q \offer[p]$.  Notice that
  $q\offer[\emptyset] \equiv \true$.
\end{definition}

\begin{note}
  In the following, we assume, for any $B_i = (Q_i, P_i, \goesto, \offer)$
  with $i \in [1,n]$, that $\{P_i\}_{i=1}^n$ are pairwise disjoint (\ie $i
  \neq j$ implies $P_i \cap P_j = \emptyset$) and $P \bydef{=}
  \bigcup_{i=1}^n P_i$.

  To avoid excessive notation, here and in the rest of the paper, we drop
  the indices on $\goesto$ and $\offer$, as they can always be
  unambiguously deduced from the corresponding state variables.
\end{note}

Let $P$ be a set of ports.  We denote $\fire{P} \bydef{=}
\setdef{\fire{p}}{p \in P}$ and $\inhib{P} \bydef{=}
\setdef{\inhib{p}}{p \in P}$.  We call the elements of $\act{P}$,
$\fire{P}$ and $\inhib{P}$ respectively {\em activation}, {\em firing}
and {\em negative port typings}.

\begin{definition}
  \label{defn:interaction}
  An {\em interaction} is a subset $a \subseteq \act{P} \cup \fire{P} \cup
  \inhib{P}$.

  For a given interaction $a$, we define the following sets of ports:
  \begin{itemize}
    \item $\actsup{a} \bydef{=} a \cap P$, the {\em activation support} of
      $a$,
    \item $ \firesup{a} \bydef{=} \setdef{p \in P}{\fire{p} \in a}$, the
      {\em firing support} of $a$,
    \item $\negsup{a} \bydef{=} \setdef{p \in P}{\inhib{p} \in a}$, the
      {\em negative support} of $a$.
  \end{itemize}
\end{definition}

\begin{definition}
  \label{defn:composition}
  Let $B_i = (Q_i, P_i, \goesto, \offer)$, with $i \in [1,n]$ and $P =
  \bigcup_{i=1}^n P_i$, be a set of component behaviours.  Let $\gamma
  \subseteq 2^{\act{P} \cup \fire{P} \cup \inhib{P}}$ be a set of
  interactions.  The composition of $\{B_i\}_{i=1}^n$ with $\gamma$ is
  a behaviour $\gamma(B_1,\dots,B_n) \bydef{=} (Q, P, \goesto,
  \offer)$ with
  \begin{itemize}
  \item the set of states $Q = \prod_{i=1}^n Q_i$\mdash the cartesian
    product of the sets of states $Q_i$,
  \item the strongest (\ie inductively defined) offer predicate $\offer$
    satisfying the rules, for each $i \in [1,n]$,
    \begin{equation}
      \label{eq:rule:offer}
      \derrule{q_i \offer p}{q_1\dots q_n \offer p}
    \end{equation}
    (recall that the sets of ports $P_i$ are pairwise disjoint),
  \item the minimal transition relation $\goesto$ satisfying the
    rule
    \begin{equation}
    \label{eq:rule:trans}
     \renewcommand{\arraystretch}{2}
     \derrule[4]{
       a \in \gamma
       & \Big\{q_i\longgoesto[\firesup{a}\cap P_i] q_i'\Big\}_{i=1}^n
       & \Big\{q_i \offer (\actsup{a} \cap P_i)\Big\}_{i=1}^n
       & \Setdef{q_i \noffer p}{p \in \negsup{a} \cap P_i}_{i=1}^n
     }
     {q_1\dots q_n \longgoesto[\firesup{a}] q_1'\dots q_n'}\,.
   \end{equation}
  \end{itemize}
\end{definition}

\begin{wrapfigure}{r}{0.24\textwidth}
  \centering
  \input{nonflat-new.pspdftex}
  \mycaption{Flat composed system equivalent to the one shown in
    \fig{nonflat:composed}}
  \label{fig:nonflat:new}
\end{wrapfigure}

Taking on the \ex{nonflat} from the introduction, a flat composition of
$B_1$, $B_2$ and $B_3$ equivalent to that of \fig{nonflat:composed} in the
semantics of \defn{composition} is shown in \fig{nonflat:new} on the right.
This representation follows the classical BIP approach with the exception
of the priority, whereof the semantics is defined in terms of the offer
predicate.  In terms of \defn{composition}, this is translated by taking
$\gamma = \{\fire{p}\,\inhib{r}, \fire{q}, \fire{s}, \fire{r}\,\fire{t}\}
\subseteq 2^{\act{P} \cup \fire{P} \cup \inhib{P}}$.

BIP composition operators, consisting of an interaction and a priority
model, can be given new operational semantics in terms of the offer
predicate: the semantics of the interaction model composition remains the
same \eq{transsem}, whereas the rule for priority becomes
\begin{equation}
  \label{eq:newprisem}
  \derrule[2]{
    q \longgoesto[a] q' &
     \Setdef{q \noffer[a']}{a \less a'}
  }{
    q \longgoesto[a]_\less q'
  }\,.
\end{equation}
Clearly, any combination of BIP interaction and priority models can be
represented by an extended interaction model $\gamma \subseteq 2^{\act{P}
  \cup \fire{P} \cup \inhib{P}}$.  A priority $a \less p_1\dots p_n$ is
translated into $\{\fire{a}\,\inhib{p_1}, \dots, \fire{a}\,\inhib{p_n}\}$
(here $\fire{a}$ is a shorthand for the set of firing ports corresponding
to ports in $a$).  In general, when several inhibitors are defined for the
same interaction, that is $a \less p_1^i\dots p_{n_i}^i$, for $i \in
[1,m]$, this translates into
$\setdef{\fire{a}\,\inhib{p^1_{k_1}}\,\dots\,\inhib{p^m_{k_m}}}{k_i \in
  [1,n_i]}$.

It is important to observe that, as stated by \lem{nofiring} below, the
rule \eq{rule:trans} in \defn{composition} implies that any interaction $a
\in \gamma$ such that $\firesup{a} = \emptyset$ does not have any impact on
the composed system.

\begin{lemma}
  \label{lem:nofiring}
  Let $\gamma_1, \gamma_2 \subseteq 2^{\act{P} \cup \fire{P} \cup
    \inhib{P}}$ be two sets of interactions and denote
  $\symmdiff{\gamma_1}{\gamma_2} \bydef{=} (\gamma_1 \setminus \gamma_2)
  \cup (\gamma_2 \setminus \gamma_1)$ their symmetric difference.  If
  $\firesup{a} = \emptyset$, for all $a \in \symmdiff{\gamma_1}{\gamma_2}$,
  then $\gamma_1(B_1,\dots,B_n) = \gamma_2(B_1,\dots,B_n)$.
\end{lemma}
\begin{proof}
  It is easy to see that $\gamma_1(B_1,\dots,B_n)$ and
  $\gamma_2(B_1,\dots,B_n)$ behaviours can differ only in their respective
  transition relations $\goesto$.

  Application of the rule \eq{rule:trans} in \defn{composition} to an
  interaction with empty firing support generates a transition $q_1\dots
  q_n \longgoesto[\firesup{a} = \emptyset] q_1\dots q_n$. As mentioned in
  the opening of this section, we assume that the self-loop transition
  labelled by an empty set is enabled in all states. Therefore, the above
  transition is present in both $\gamma_1(B_1,\dots,B_n)$ and
  $\gamma_2(B_1,\dots,B_n)$.  By the assumption of the lemma, all
  interactions with non-empty firing support belong to $\gamma_1 \cap
  \gamma_2$.  Hence all transitions labelled with non-empty interactions
  also appear in both $\gamma_1(B_1,\dots,B_n)$ and
  $\gamma_1(B_1,\dots,B_n)$.
\end{proof}

\begin{lemma}
  \label{lem:minimal}
  Let $\gamma_1 \subseteq 2^{\act{P} \cup \fire{P} \cup \inhib{P}}$ be a
  set of interactions, $\gamma_2 = \gamma_1 \cup \{a\}$, with $a \subseteq
  \act{P} \cup \fire{P} \cup \inhib{P}$, such that there is an interaction
  $b \in \gamma_1$, $b \subseteq a$ and $\firesup{b}=\firesup{a}$. Then
  $\gamma_1(B_1,\dots,B_n) = \gamma_2(B_1,\dots,B_n)$.
\end{lemma}
\begin{proof}
  According to rule \eq{rule:trans} any transition generated by the
  interaction $a$ can also be generated by the interaction $b$. Thus,
  interaction $a$ does not impact the behaviour of the composed system, and
  $\gamma_1(B_1,\dots,B_n) = \gamma_2(B_1,\dots,B_n)$.
\end{proof}
  

\section{Algebra extensions} 
\label{sec:extension}

In \secn{semantic}, we have replaced the classical BIP combination of
interaction and priority models with an extended interaction model with
ports of three types: firing, activation and negative.\footnote{
  Only firing and negative ports are necessary to define classical BIP
  composition operators.  Activation ports allow for a full correspondence
  with $\sB[\act{P},\fire{P}]$ Boolean constraints.  This correspondence
  and an expressivity study are given in \cite{BliSif11-constraints-sc}.
}  We can now extend other algebras used for the glue representation.

We start by considering the extension of the Algebra of Interactions,
$\ai(P)$.  Recall that $x \simeq y$ iff $\intsem{x} = \intsem{y}$.  As a
simple corollary of the results in \cite{BliSif08-express-concur},
$\intsem{x} = \intsem{y}$ is equivalent to $\intsem{x}(\bB) =
\intsem{y}(\bB)$, for any finite family $\bB$ of behaviours.

Below we will consider $\ai(\act{P} \cup \fire{P} \cup \inhib{P})$ with the
latter definition of term equivalence: two terms $x,y \in \ai(\act{P} \cup
\fire{P} \cup \inhib{P})$ are equivalent iff $\intsem{x}(\bB) =
\intsem{y}(\bB)$ (in terms of \defn{composition}), for any finite family
$\bB$ of behaviours.  In general, we define equivalence as follows.

\begin{definition}
Let $\cA(P)$ be an algebra, $\intsem{\cdot} : \cA(P) \rightarrow 2^{2^P}$.
Two terms $x, y \in \cA(P)$ are {\em equivalent} $x \sim y $ iff, for any
finite family $\bB$ of behaviours, $\intsem{x}(\bB) = \intsem{y}(\bB)$ (in
terms of \defn{composition}).
\end{definition}

\begin{note}
  Clearly $\sim$ is weaker than $\simeq$.
\end{note}

We are now in position to similarly extend the other algebras.  The
interaction semantics of the causal interaction trees $\aisem{\cdot}:
\ct(P) \rightarrow \ai(P)$ is transposed without any change to
$\aisem{\cdot}: \ct(\act{P} \cup \fire{P} \cup \inhib{P}) \rightarrow
\ai(\act{P} \cup \fire{P} \cup \inhib{P})$.  Similarly, the functions
$\tau: \ac(P) \rightarrow \ct(P)$ and $\sigma: \ct(P) \rightarrow \ac(P)$
are transposed identically to $\ac(\act{P} \cup \fire{P} \cup \inhib{P})$
and $\ct(\act{P} \cup \fire{P} \cup \inhib{P})$.  The same goes for the
mapping $R(t)$ associating to a causal interaction tree $t \in \ct(P)$ the
corresponding system of causal rules \cite{BliSif10-causal-fmsd}.  The only
difference is that, in $\cru(\act{P} \cup \fire{P} \cup \inhib{P})$ we
introduce the following additional axiom: $\fire{p} \Rightarrow \act{p}$,
for all $p \in P$.

\begin{proposition} \label{prop:congruence}
  The equivalence relation $\sim$ on $\ct(\act{P} \cup \fire{P} \cup
  \inhib{P})$ is a congruence.
\end{proposition}
\begin{proof}[Sketch of the proof]
  The proof is the same as for $\ct(P)$ \cite{BliSif10-causal-fmsd}.  For
  any two trees $t_1, t_2 \in \ct(\act{P} \cup \fire{P} \cup \inhib{P})$
  and for any context $C(z) \in \ct(\act{P} \cup \fire{P} \cup \inhib{P}
  \cup \{z\})$, we have to show that the equivalence $t_1 \sim t_2$ implies
  $C(t_1/z) \sim C(t_2/z)$, where $C(t_i/z)$ is the tree obtained, by
  replacing in $C(z)$ all occurrences of $z$ by $t_i$.  Since the semantics
  $\ct$ is compositional, structural induction on the context $C(z)$ proves
  the proposition.
\end{proof}

The first consequence of the above extension is that, rather than extending the
existing graphical representation of connectors, it can be directly used in its present form to
express priorities and activation conditions (the use of the offer
predicate in the positive premises of the rule \eq{rule:trans}) by adding a
trivalued attribute to ports: {\em firing}, {\em activation} and {\em
  negative}.  It is important to observe the difference between, on one
hand, adding an attribute to ports and, on the other hand, modifying the
typing operator ({\em synchron} vs.\ {\em trigger} typing), since the
latter is applied at each level of the connector hierarchy, whereas the
former is applied to ports, that is only at the leaves of the connector.


\subsection{Refinement of the extension}
\label{sec:refinement}

When we apply $x, y \in \ai(\act{P} \cup \fire{P} \cup \inhib{P})$ to
compose behaviour with operational semantics of \defn{composition},
$\intsem{x}(\bB) = \intsem{y}(\bB)$ does not imply $x = y$.  $\ai$ axioms
are not complete (although still sound) with respect to $\sim$, since this
equivalence is weaker than $\simeq$. Consequently, on $\ct(\act{P} \cup
\fire{P} \cup \inhib{P})$, $\sim$ is also weaker than $\simeq$.

\begin{figure*}
  \centering
  \begin{subfigure}[b]{1em}
    \centering
    \begin{tikzpicture}
      \node (P) at (0,0) {$\inhib{p}$};
      \node (Q) at (0,-1) {$\fire{q}$};
      \node at (0,-1.5) {};
      \path[->]
      (P) edge (Q);
    \end{tikzpicture}
    \caption{}
    \label{fig:equiv:trees:simple}
  \end{subfigure}
  \hspace{1cm} 
  \begin{subfigure}[b]{2em}
    \centering
    \begin{tikzpicture}
      \node at (0,-1) {$\inhib{p}\, \fire{q}$};
      \node at (0,-2) {};
    \end{tikzpicture}
    \caption{}
    \label{fig:equiv:trees:simple:down}
  \end{subfigure}
  \hspace{2cm} 
  \begin{subfigure}[b]{25mm}
    \centering
    \begin{tikzpicture}
      \node (P) at (1,0) {$\fire{p}$};
      \node (Q) at (1,-1) {$\inhib{q}$};      
      \node (R) at (0,-2) {$\fire{r}$};
      \node (S) at (2,-2) {$\fire{s}$};
      \path[->]
      (P) edge (Q) 
      (Q) edge (R) edge (S);      
    \end{tikzpicture}
    \caption{}
    \label{fig:equiv:trees:adv}
  \end{subfigure}
  \hspace{1cm}
  \begin{subfigure}[b]{25mm}
    \centering
    \begin{tikzpicture}
      \node (P) at (1,0) {$\fire{p}$};
      \node (R) at (0,-2) {$\inhib{q}\, \fire{r}$};
      \node (S) at (2,-2) {$\inhib{q}\, \fire{s}$};
      \path[->]
      (P) edge (R) edge (S);      
    \end{tikzpicture}
    \caption{}
    \label{fig:equiv:trees:adv:down}
  \end{subfigure}
  \mycaption{Two pairs of equivalent trees: (\subref{fig:equiv:trees:simple}),
    (\subref{fig:equiv:trees:simple:down}) and (\subref{fig:equiv:trees:adv}),
    (\subref{fig:equiv:trees:adv:down})}
  \label{fig:equiv:trees}
\end{figure*}

\begin{example}
Let $P = \{p,q,r,s\}$ and consider the $\ct(\act{P} \cup \fire{P} \cup
\inhib{P})$ trees shown in \fig{equiv:trees}.  The interaction semantics of
the tree in \fig{equiv:trees:simple} is $\intsem{\inhib{p} \rightarrow
  \fire{q}} = \{\inhib{p},\, \inhib{p}\,\fire{q}\}$.  However, the
interaction $\inhib{p}$ does not contain any firing ports.  Therefore, as
mentioned above (\lem{nofiring}), it does not influence component
synchronisation and we have $\inhib{p} \rightarrow
\fire{q}\ \sim\ \inhib{p}\,\fire{q}$ (\cf \fig{equiv:trees:simple:down}).

The causal interaction tree in \fig{equiv:trees:adv} also defines a
redundant interaction.  Indeed,
\[
  \intsem{
    \fire{p} \rightarrow \inhib{q} \rightarrow (\fire{r} \oplus \fire{s})
  }
  \ =\ 
  \left\{
  \fire{p},\, 
  \fire{p}\, \inhib{q}\,,\, 
  \fire{p}\, \inhib{q}\, \fire{r}\,,\, 
  \fire{p}\, \inhib{q}\, \fire{s}\,,\, 
  \fire{p}\, \inhib{q}\, \fire{r}\, \fire{s}\,
  \right\}\,.
\]

Although the interaction $\fire{p}\, \inhib{q}$ does contain a firing port
$\fire{p}$, it is redundant (\lem{minimal}). We conclude, therefore, that
the causal interaction trees in \fig{equiv:trees:adv} and
\fig{equiv:trees:adv:down} are equivalent, since
\[
  \intsem{
    \fire{p} \rightarrow (\inhib{q}\,\fire{r} \oplus \inhib{q}\,\fire{s})
  }
  \ =\ 
  \left\{
  \fire{p},\, 
  \fire{p}\, \inhib{q}\, \fire{r}\,,\, 
  \fire{p}\, \inhib{q}\, \fire{s}\,,\, 
  \fire{p}\, \inhib{q}\, \fire{r}\, \fire{s}\,
  \right\}\,.
\]
\end{example}

The above example illustrates the idea that the nodes of causal interaction
trees, which do not contain firing ports, can be ``pushed'' down the tree.

Another notable case leading to redundant interactions corresponds to trees
containing {\em contradictory port typings}.  For example, either of the
two equivalent trees $\inhib{p} \rightarrow \fire{p}$ and
$\inhib{p}\,\fire{p}$ authorises the interaction $\inhib{p}\,\fire{p}$.
However, when considered in the context of the rule \eq{rule:trans}, this
interaction generates two conflicting premises $q_i \goesto[p] q_i'$ and
$q_i \noffer[p]$.  Thus, this instance of the rule \eq{rule:trans} does not
authorise any transitions and the interaction $\inhib{p}\,\fire{p}$ can be
safely discarded.  This example corresponds to the additional axiom
$\fire{p} \Rightarrow \act{p}$ imposed in \cite{BliSif11-constraints-sc} on
the Boolean formul\ae\ in $\sB[\act{P},\fire{P}]$.  Similarly, redundant
interactions appear when a tree contains other distinct port typings of the
same port: $\act{p}$ and $\inhib{p}$, generating conflicting premises $q_i
\offer[p]$ and $q_i \noffer{p}$; $\act{p}$ and $\fire{p}$, whereof the
former generates the premise $q_i \offer[p]$ redundant alongside the
premise $q_i \goesto[p] q_i'$ generated by the latter.

Below, we provide a set of axioms reducing interaction redundancy.  We
enrich axioms for $\ct(\act{P} \cup \fire{P} \cup \inhib{P})$ from
\cite{BliSif10-causal-fmsd} by adding some new ones, specific for the
trivalued port attribute.

\begin{axioms}
  \begin{enumerate}
  \item \label{ax:nodes} For all $p \in P$ and $a \subseteq \act{P} \cup
    \fire{P} \cup \inhib{P}$ such that $a \neq \emptyset$,
    \begin{enumerate}
    \item $a \cdot 0 = 0$,
    \item $a \cdot 1 = a$, for $a \neq 0$,
    \item $\fire{p}\cdot\act{p} = \fire{p}$ (\cf the additional axiom
      $\fire{p} \Rightarrow \act{p}$ in $\cru(\act{P} \cup \fire{P} \cup
      \inhib{P})$),
    \item $\fire{p}\cdot\inhib{p} = \act{p}\cdot\inhib{p} = 0$.
    \end{enumerate}
  \item \label{ax:par} Parallel composition, `$\oplus$', is associative,
    commutative, idempotent, and its identity element is $0$.
  \item \label{ax:zero:leaf} $a \rightarrow 0 = a$, for all $a
    \subseteq \act{P} \cup \fire{P} \cup \inhib{P}$.
  \item \label{ax:zero:node} $0 \rightarrow t = 0$, for all $t \in
    \ct(\act{P} \cup \fire{P} \cup \inhib{P})$.
  \item \label{ax:pushdown:node} $c \rightarrow a \rightarrow b 
    \rightarrow t = c \rightarrow ab \rightarrow t$ for all $a,b,c
    \subseteq \act{P} \cup \fire{P} \cup \inhib{P}$, such that 
    $\firesup{a} = \emptyset$, and $t \in \ct(\act{P} \cup \fire{P}
    \cup \inhib{P})$.
  \item \label{ax:pushdown:port} $ap \rightarrow b = ap
    \rightarrow bp$ for all $a,b \subseteq \act{P} \cup \fire{P} \cup
    \inhib{P}$, $p \in \act{P} \cup \fire{P} \cup  \inhib{P} $.
  \item \label{ax:relation:operators}
    $a \rightarrow (t_1 \oplus t_2) = a \rightarrow t_1\ \oplus\ a
    \rightarrow t_2$, for all $a \subseteq \act{P} \cup
    \fire{P} \cup \inhib{P}$, $t_1, t_2 \in
    \ct(\act{P} \cup \fire{P} \cup \inhib{P})$.
  \end{enumerate}
\end{axioms}

\axs{nodes} equalise redundant interactions due to contradictory port
typings, whereas \ax{pushdown:node} eliminates the nodes with empty firing
support.  \axs{par}, \ref{ax:zero:leaf}, \ref{ax:zero:node} and
\ref{ax:relation:operators} are the same as in
\cite{BliSif10-causal-fmsd}.\footnote{
  The two remaining axioms from \cite{BliSif10-causal-fmsd} are replaced by
  Lemmas~\ref{lem:nofiring:leaf} and \ref{lem:one:node} in this paper.
}

\begin{proposition}
  The above axiomatisation is sound with respect to $\sim$.
\end{proposition}
\begin{proof}
Since, by \prop{congruence}, the equivalence relation $\sim$ is a
congruence, it is sufficient to show that all the axioms respect $\sim$.
This is proved by verifying that the semantics for left and right sides
coincide.

\axs{par}, \ref{ax:zero:leaf}, \ref{ax:zero:node} and
\ref{ax:relation:operators} are the same as in \cite{BliSif10-causal-fmsd}.
Hence, their respective left- and right-hand sides are related by $\simeq$,
which is stronger than $\sim$.  \ax{nodes}(a) and \ax{nodes}(b) are
trivial. \ax{nodes}(c) is a consequence of \lem{minimal}.  In the
\ax{nodes}(d), both pairs $p$ and $\inhib{p}$, and $\fire{p}$ and
$\inhib{p}$ produce conflicting premises in the rule \eq{rule:trans} and,
therefore, do not generate any transitions.  For the \ax{pushdown:node}, we
have
\begin{align}
  \intsem{c \rightarrow a \rightarrow b \rightarrow t} 
  &= \{c,\, a\,c,\, a\,b\,c\} \cup \setdef{a\,b\,c\,a_2}{a_2 \in \intsem{t}}\\
  \intsem{c \rightarrow ab \rightarrow t} 
  &= \{c,\, a\,b\,c\} \cup \setdef{a\,b\,c\,a_2}{a_2 \in \intsem{t}}
\end{align}
The only difference between the interaction semantics of the two trees is
the interaction $ac$.  However, any transition authorised by the rule
\eq{rule:trans} with this interaction is also authorised with interaction
$c$, since $\firesup{a} = \emptyset$ (\lem{minimal}).  Hence, the composed
systems coincide.

For the \ax{pushdown:port}, we have $\intsem{ap \rightarrow b} = \{ap,\,
abp\} = \intsem{ap \rightarrow bp}$.  Thus $ap \rightarrow b \simeq ap
\rightarrow bp$, which implies $ap \rightarrow b \sim ap \rightarrow bp$.
\end{proof}

Notice that our axiomatisation is not complete.  For instance, the equivalence $p\rightarrow q \oplus q \rightarrow p \sim p \oplus q$ cannot be derived from the axioms.

\begin{lemma} \label{lem:nofiring:leaf} 
  For all $a,b \subseteq \act{P} \cup \fire{P} \cup \inhib{P}$, such that
  $\firesup{b} = \emptyset$, holds the equality $a \rightarrow b = a$.
\end{lemma}
\begin{proof}
$a \rightarrow b = a \rightarrow b \rightarrow 0 \rightarrow 0 = a
  \rightarrow b \cdot 0 \rightarrow 0 = a \rightarrow 0 \rightarrow 0 = a$
  (\axss{zero:leaf}{pushdown:node})
\end{proof}

\begin{lemma} \label{lem:one:node} 
  For all $a \subseteq \act{P} \cup \fire{P} \cup \inhib{P}$ and $t \in
  \ct(\act{P} \cup \fire{P} \cup \inhib{P})$, holds the equality $a
  \rightarrow 1 \rightarrow t = a \rightarrow t$.
\end{lemma}
\begin{proof}
If $t = 0$ the statement of this lemma is a special case of
\lem{nofiring:leaf} with $b = 1$.  If $t \neq 0$ it can be represented as a
parallel composition of non-zero trees $t = \bigoplus_{i=1}^{n} r_i
\rightarrow t_i$, with $r_i \subseteq \act{P} \cup \fire{P} \cup
\inhib{P}$.  By \axs{pushdown:node} and \ref{ax:relation:operators}, we
have
\[
a \rightarrow 1 \rightarrow t 
\ =\ 
\bigoplus_{i=1}^{n} (a \rightarrow 1 \rightarrow r_i \rightarrow t_i)
\ =\ 
\bigoplus_{i=1}^{n} (a \rightarrow r_i \rightarrow t_i)
\ =\ 
a \rightarrow \bigoplus_{i=1}^{n} (r_i \rightarrow t_i) 
\ =\ 
a \rightarrow t\,.
\]
\end{proof}

\begin{lemma} \label{lem:pushdown:node}
  For all $a,b_i,c \subseteq \act{P} \cup \fire{P} \cup \inhib{P}$, such
  that $\firesup{a} = \emptyset$ and $t_i \in \ct(\act{P} \cup \fire{P}
  \cup \inhib{P})$, holds the equality 
\[
c \rightarrow a \rightarrow
\bigoplus_{i=1}^{n} (b_i \rightarrow t_i) = c \rightarrow
\bigoplus_{i=1}^{n} (ab_i \rightarrow t_i)\,.
\]
\end{lemma}
\begin{proof}
As above, applying \axs{pushdown:node} and \ref{ax:relation:operators}, we
have
\[
c \rightarrow a \rightarrow \bigoplus_{i=1}^{n} (b_i \rightarrow t_i)
\ =\
\bigoplus_{i=1}^{n} (c \rightarrow a \rightarrow b_i \rightarrow t_i) 
\ =\
\bigoplus_{i=1}^{n} (c \rightarrow ab_i \rightarrow t_i) 
\ =\
c \rightarrow \bigoplus_{i=1}^{n} (ab_i \rightarrow t_i)\,.
\]
\end{proof}

\begin{definition} \label{defn:tree:normal}
  A causal interaction tree $t \in \ct(\act{P} \cup \fire{P} \cup
  \inhib{P})$ is in \emph{normal form} if it satisfies the following
  properties:
  \begin{enumerate}
  \item All nodes of $t$ except roots have non-empty firing support.
  \item There are no causal dependencies between the same typing of the
    same port in $t$, that is for any causal chain $a \rightarrow \dots
    \rightarrow b$ within $t$, we have $a \cap b = \emptyset$.
  \item There are no causal dependencies between different port typings of
    the same port in $t$, other than dependencies of the form $ap
    \rightarrow \dots \rightarrow b \fire{p}$, where $a,b \subseteq \act{P}
    \cup \fire{P} \cup \inhib{P}$, $p \in P$.
  \end{enumerate}
\end{definition}

\begin{proposition}[Normal form for causal interaction trees]
  \label{prop:normal:tree}
  Every causal interaction tree $t \in \ct(\act{P} \cup \fire{P} \cup
  \inhib{P})$ has a normal form $t = \tilde{t} \in \ct(\act{P} \cup
  \fire{P} \cup \inhib{P})$.
\end{proposition}

\begin{proof}
  Consider $t \in \ct(\act{P} \cup \fire{P} \cup \inhib{P})$.  We start by
  computing $t_1 = t$ with all nodes, except potentially the roots, having
  non-empty firing support.  
  
  Let $a$ be a non-root node of $t$ with $\firesup{a} = \emptyset$, such
  that the tree $s$ rooted in $a$ does not have any nodes with empty firing
  support.  If $s$ is empty, that is $a$ is a leaf then remove $a$ from the
  tree (\lem{nofiring:leaf}).  Otherwise, let $c$ be the parent of $a$, which
  exists since $a$ is not a root and move the parallel composition operator
  up using \ax{relation:operators}:
  \begin{equation}
    \label{eq:plusup}
    c \rightarrow 
    \left((a \rightarrow s) \oplus \bigoplus_{i=1}^{n} t_i\right)
    =
    (c \rightarrow a \rightarrow s) \oplus 
    \left(\bigoplus_{i=1}^{n} c \rightarrow t_i\right)\,.
  \end{equation}
  The sub-tree $s$ can be further decomposed as $s = \bigoplus_{i=1}^n (b_i
  \rightarrow s_i)$, so, by \lem{pushdown:node}, we have
  \begin{equation}
    \label{eq:pushdown}
    c \rightarrow a \rightarrow s \
    =\ c \rightarrow a \rightarrow \bigoplus_{i=1}^n (b_i \rightarrow s_i)\
    =\ c \rightarrow \bigoplus_{i=1}^{n} (ab_i \rightarrow s_i)\,.
  \end{equation}
  Each of nodes $ab_i$ has non-empty firing support, since $\firesup{b_i} =
  \emptyset$ by the choice of $a$.  Substituting \eq{pushdown} into
  \eq{plusup} and applying \ax{relation:operators}, we obtain
  \[
    \left(c \rightarrow \bigoplus_{i=1}^{n} (ab_i \rightarrow s_i)\right) 
    \oplus 
    \left(\bigoplus_{i=1}^{n} c \rightarrow t_i\right) 
    =
    c \rightarrow \left(\left(\bigoplus_{i=1}^{n} ab_i \rightarrow s_i\right) 
    \oplus 
    \bigoplus_{i=1}^{n} t_i\right)\,.
  \]

  In the resulting tree, there is one node with empty firing support less
  than in $t$.  Hence, repeating this procedure as long as there are such
  nodes, we will compute a tree $t_1$, where all nodes except roots have
  non-empty firing support.  This computation is confluent, since the order
  is irrelevant among causally independent nodes, whereas among causally
  dependent ones it is fixed by the algorithm.
  
  Consider a causal chain $a\tilde{p} \rightarrow \dots \rightarrow
  b\hat{p}$ within $t_1$, with $\tilde{p}$ and $\hat{p}$ being two typings
  of the same port.  If $\tilde{p} = \act{p}$ and $\hat{p} = \fire{p}$,
  there is nothing to do, since such dependencies are allowed by
  \defn{tree:normal}.  Otherwise, we propagate $\tilde{p}$ down by applying
  \ax{pushdown:port}:
  \[
    a\tilde{p} \rightarrow c_1 \rightarrow 
    \dots \rightarrow c_k \rightarrow b\hat{p} 
    \ =\ 
    a\tilde{p} \rightarrow c_1\tilde{p} \rightarrow 
    \dots \rightarrow c_k \rightarrow b\hat{p}
    =\ \dots\ =
    a\tilde{p} \rightarrow c_1\tilde{p} \rightarrow 
    \dots \rightarrow c_k\tilde{p} \rightarrow b\hat{p}\tilde{p}\,.
  \]
  {\bf Case 1:} $\tilde{p} = \hat{p}$ or $\tilde{p} = \fire{p}$ and
  $\hat{p} = \act{p}$. We apply \axs{nodes}(c) and \ref{ax:pushdown:port}:
  \[
    a\tilde{p} \rightarrow c_1\tilde{p} \rightarrow 
    \dots \rightarrow c_k\tilde{p} \rightarrow b\hat{p}\tilde{p} 
    \ =\ 
    a\tilde{p} \rightarrow c_1\tilde{p} \rightarrow 
    \dots \rightarrow c_k\tilde{p} \rightarrow b\tilde{p}
    \ =\ 
    a\tilde{p} \rightarrow c_1 \rightarrow 
    \dots \rightarrow c_k \rightarrow b\,.
  \]
  {\bf Case 2:} $\tilde{p} \neq \hat{p}$ and either $\tilde{p} = \inhib{p}$
  or $\hat{p} = \inhib{p}$.  We apply \axs{nodes}(d), \ref{ax:zero:leaf}
  and \ref{ax:pushdown:port}:
  \begin{multline*}
    a\tilde{p} \rightarrow c_1\tilde{p} \rightarrow
    \dots \rightarrow c_k\tilde{p} \rightarrow b\hat{p}\tilde{p} 
   \ =\ 
    a\tilde{p} \rightarrow c_1\tilde{p} \rightarrow 
    \dots \rightarrow c_k\tilde{p} \rightarrow 0\ = \\
   \ =\ 
    a\tilde{p} \rightarrow c_1 \rightarrow 
    \dots \rightarrow c_k \rightarrow 0
   \ =\ 
    a\tilde{p} \rightarrow c_1 \rightarrow \dots \rightarrow c_k\,.
  \end{multline*}
    
  To compute $\tilde{t}$, we apply this transformation to all relevant
  causal chains within $t$.
\end{proof}

\begin{definition} \label{defn:conn:normal}
  An $\ac(\act{P} \cup \fire{P} \cup \inhib{P})$ connector is in
  \textit{normal form} if the following conditions hold.
  \begin{enumerate}
  \item Nodes at every hierarchical level of the connector, except the
    bottom one, have at least one trigger.
  \item Each node at the bottom hierarchical level, is a strong
    synchronisation of pairwise distinct ports.
  \item Every node at the bottom hierarchical level, without firing ports,
    has only triggers as ancestors.
  \end{enumerate}
\end{definition}

\begin{corollary}[Normal form for connectors] 
  Every connector $x \in \ac(\act{P} \cup \fire{P} \cup \inhib{P})$ has an
  equivalent normal form $x \sim \tilde{x} \in \ac(\act{P} \cup \fire{P}
  \cup \inhib{P})$.
\end{corollary}
\begin{proof}[Sketch of the proof]
  Given a connector $x$, let $t = \tau(x)$ be the equivalent causal
  interaction tree and $\tilde{t} = t$ its normal form.  Put $\tilde{x} =
  \sigma(\tilde{t})$.  Since both $\sigma$ and $\tau$ preserve $\sim$, we
  have $\tilde{x} \sim x$.  Normality of $\tilde{x}$ is a direct
  consequence of that of $\tilde{t}$ and the definition \eq{treecon} of
  $\sigma$.
\end{proof}

\begin{proposition}
  \label{prop:normal:rules}
  Any causal interaction tree $t \in \ct(\act{P} \cup \fire{P} \cup
  \inhib{P})$ can be represented by a system of causal rules with only
  firing ports as effects, \ie having only rules of the form $\fire{p}
  \Rightarrow C$, where $C$ is a DNF Boolean formula on $\fire{P} \cup
  \act{P}$ without negative firing variables.
\end{proposition}
\begin{proof}
Applying the transformation $R : \ct(P) \rightarrow \cru(P)$ defined in
\secn{transformations} to a tree $t \in \ct(P)$, gives a system of causal
rules of the form $p \Rightarrow C$, where $C$ is a DNF Boolean formula and
each monomial is a conjunction of the nodes on the way from a root of $t$
to $p$ (some prefix in $t$ leading to $p$, excluding $p$).

We define the transformation $\tilde{R}:\ct(\act{P} \cup \fire{P} \cup
\inhib{P}) \rightarrow \cru(\act{P} \cup \fire{P} \cup \inhib{P})$, by
putting
\begin{equation}
  \label{eq:trees2rules1}
  \tilde{R}(t) \bydef{=} \{p \Rightarrow c_p(t)\}_{p\in \fire{P}\cup\{\true\}}\,,
\end{equation}
that is we omit causal rules for port variables in $\act{P} \cup \inhib{P}$
(in \eq{trees2rules1}, the set of rules is indexed by $p \in
\fire{P}\cup\{\true\}$ as opposed to $p \in P\cup\{\true\}$ in
\eq{trees2rules}).  To prove the equivalence $t \sim \tilde{R}(t)$ it is
sufficient to show $\tilde{R}(t) \sim R(t)$.

$\tilde{R}(t)$ has less constraints than $R(t)$.  Hence, it allows more
interactions.  Let $a \in \intsem{\tilde{R}(t)} \setminus \intsem{R(t)}$,
i.e. there exists $ p \in \act{P} \cup \inhib{P}$, such that $p \in a$ and
the rule $p \Rightarrow C_1$ is violated by $a$.  Let $\tilde{a} = a
\setminus p$.

Assume $\tilde{a} \notin \intsem{\tilde{R}(t)}$, \ie there exists $\fire{q}
\in \fire{P}$ and a rule $(\fire{q} \Rightarrow C_2) \in \tilde{R}(t)$,
such that $\fire{q} \in \tilde{a}$ and the rule $\fire{q} \Rightarrow C_2$
is violated by $\tilde{a}$.  This rule is not violated by $a$. Hence $C_2 =
pC_2'$ and, consequently, $p$ lies on all prefixes in $t$, leading to
$\fire{q}$.  $a \in \intsem{\tilde{R}(t)}$, $\fire{q} \in \tilde{a}
\subseteq a$, thus there is at least one prefix in $t$, leading to
$\fire{q}$ and contained in $a$. As $p$ lies on this prefix, the rule $(p
\Rightarrow C_1)$ is satisfied by $a$, contradicting the conclusion above.
Therefore our assumption is wrong and $\tilde{a} \in
\intsem{\tilde{R}(t)}$.

Since $\tilde{a} \in \intsem{\tilde{R}(t)}$ and $\firesup{\tilde{a}} =
\firesup{a}$, we have, by \lem{minimal}, $\intsem{\tilde{R}(t)}(\bB) =
(\intsem{\tilde{R}(t)} \setminus \{a\})(\bB)$ for any family of behaviours
$\bB$.  Thus, for all $a \in \intsem{\tilde{R}(t)} \setminus
\intsem{R(t)}$, there exists $\tilde{a} \varsubsetneq a$, such that
$\tilde{a} \in \intsem{\tilde{R}(t)}$ and $\firesup{\tilde{a}} =
\firesup{a}$.  By \lem{minimal}, we have $\intsem{\tilde{R}(t)}(\bB) =
\intsem{R(t)}(\bB)$ for any family $\bB$, \ie $R(t) \sim \tilde{R}(t)$.
\end{proof}


\section{Connector synthesis (example)}
\label{sec:example}

\begin{wrapfigure}[7]{r}{0.24\textwidth}
  \centering
  \vspace{-0.5\baselineskip}
  \input{control-main.pspdftex}
  \mycaption{Main module}
  \label{fig:modules}
\end{wrapfigure}
Consider a system providing some given functionality in two modes: {\em
  normal} and {\em backup}.  The system consists of four modules: the
Backup module $A$ can only perform one action $a$; the Main module $B$
(\fig{modules}) can perform an action $b$ corresponding to the normal mode
activity, it can also be
switched $on$ and $off$, as well as perform an
internal (unobservable) error transition $err$; the Monitor module $M$ is a
black box, which performs some internal logging by observing the two
actions $a$ and $b$ through the corresponding ports $a_l$ and $b_l$;
finally, the black box Controller module $C$ takes the decisions to switch
on or off the main module through the corresponding ports $on_c$ and
$off_c$, furthermore, it can perform a $test$ to check whether the main
module can be restarted.  

We want to synthesise connectors ensuring the properties below
(encoded by Boolean constraints).

\begin{itemize}
\item The main and backup actions must be logged: $\fire{a} \Leftrightarrow
  \fire{a_l}$ and $\fire{b} \Leftrightarrow \fire{b_l}$\,;
\item Only Controller can turn on the Main module: $\fire{on}
  \Leftrightarrow \fire{on_c}$\,;
\item When Controller switches off the Main module must stop operation:
  $\fire{off_c} \Rightarrow \fire{off}$ and $\fire{b} \Rightarrow
  \non{\fire{off_c}}$\,;
\item Controller can only test the execution of Backup: $\fire{test}
  \Rightarrow \fire{a}$\,;
\item Backup can only execute when Main is not possible: $\fire{a}
  \Rightarrow \inhib{b} \vee \fire{off}$\,;
\item Main can only switch off when ordered to do so or after a failure:
  $\fire{off} \Rightarrow \inhib{b} \vee \fire{off_c}$\,;
\end{itemize}

In order to compute the required glue, we take the conjunction of the above
constraints together with the \emph{progress} constraint $\fire{a} \vee
\fire{b} \vee \fire{on} \vee \fire{off} \vee \fire{test} \vee \fire{a_l}
\vee \fire{b_l} \vee \fire{off_c} \vee \fire{on_c}$ stating that at every
round some action must be taken.  In order to simplify the resulting
connectors, we also use part of the information about the behaviour of the
Main module, namely the fact that $on$, on one hand, and $b$ or $off$, on
the other, are mutually exclusive: $\act{on} \Rightarrow \inhib{b} \wedge
\inhib{off}$.  Finally, we also apply the additional axiom imposed on
Boolean constraints: $\fire{p} \Rightarrow p$.  We obtain the following
global constraint (omitting the conjunction symbol):

\begin{center}
  $\quad
  (\fire{a} \Rightarrow \fire{a_l}\ \act{a}\ \inhib{b}\
  \lor\ \fire{a_l}\ \act{a}\ \fire{off})
  (\fire{a_l} \Rightarrow \fire{a}\ \act{a_l})
  (\fire{b} \Rightarrow \fire{b_l}\ \act{b}\ \non{\fire{off_c}})
  (\fire{b_l} \Rightarrow \fire{b} \act{b_l})
  (\fire{on} \Rightarrow \fire{on_c}\ \act{on})
  (\fire{on_c} \Rightarrow \fire{on}\ \act{on_c})
  \hfill$
  \\[2pt]
  $\land\
  (\fire{off} \Rightarrow \act{off}\ \inhib{b}\
  \lor\ \fire{off_c}\ \act{off})
  (\fire{off_c} \Rightarrow \fire{off}\ \act{off_c})
  (\fire{test} \Rightarrow \fire{a}\ \act{test})$
  \\[2pt]
  $\hfill
  \land\
  (\act{on} \Rightarrow \inhib{b}\ \inhib{off})
  (\fire{a}\ \lor\ \fire{b}\ \lor\ \fire{on}\ \lor\ \fire{off}\ 
  \lor\ \fire{test}\ \lor\ \fire{a_l}\ \lor\ \fire{b_l}\ 
  \lor\ \fire{off_c}\ \lor\ \fire{on_c})\,.
  \quad$
\end{center}

Recall now that causal rules have the form $p \Rightarrow C$, where $p \in
\act{P} \cup \fire{P} \cup \inhib{P} \cup \{\true\}$ and $C$ is a DNF
Boolean formula on $\act{P} \cup \fire{P}$ without negative firing
variables.  By \prop{normal:rules}, it is sufficient to consider only the
rules with firing or $\true$ effects.  A system of causal rules is a
conjunction of such clauses.  Among the constraints above, there are two
that do not have this form: $\act{on} \Rightarrow \inhib{b} \ \inhib{off}$
and $\fire{b} \Rightarrow \fire{b_l}\ \act{b}\ \non{\fire{off_c}}$.
Rewriting them as $\inhib{on} \lor \inhib{b} \ \inhib{off}$ and
$\non{\fire{b}} \lor \fire{b_l}\ \act{b}\ \non{\fire{off_c}}$, distributing
over the conjunction of the rest of the constraints and making some
straightforward simplifications, we obtain a disjunction of three systems
of causal rules:
\begin{align*}
  \true &
  \Rightarrow \lefteqn{\fire{a}\ \inhib{b}\ \inhib{off}\ \lor\ \fire{on}}
  &&& \true &
  \Rightarrow \lefteqn{\fire{a}\ \inhib{on}\ \lor\ \fire{off}}
  &&& \true &
  \Rightarrow \lefteqn{\fire{b}\ \fire{b_l}}
  \\
  \fire{a} &\Rightarrow \fire{a_l}\ \inhib{b}
  &&& \fire{a} &\Rightarrow \lefteqn{\fire{a_l}\ \inhib{b}\ 
    \lor\ \fire{a_l}\ \fire{off}}\
  &&& \fire{a} &\Rightarrow \false
  \\
  \fire{a_l} &\Rightarrow \fire{a}
  & \fire{b} &\Rightarrow \false\quad
  & \fire{a_l} &\Rightarrow \fire{a}
  & \fire{b} &\Rightarrow \false
  & \fire{a_l} &\Rightarrow \false
  & \fire{b} &\Rightarrow \true
  \\
  \fire{on} &\Rightarrow \fire{on_c}
  &  \fire{b_l} &\Rightarrow \false
  & \fire{on} &\Rightarrow \false
  & \fire{b_l} &\Rightarrow \false
  & \fire{on} &\Rightarrow \false
  & \fire{b_l} &\Rightarrow \true
  \\
  \fire{on_c} &\Rightarrow \fire{on}
  & \fire{off} &\Rightarrow \false
  & \fire{on_c} &\Rightarrow \false
  & \fire{off} &\Rightarrow \inhib{b}\ \lor\ \fire{off_c}
  & \fire{on_c} &\Rightarrow \false
  & \fire{off} &\Rightarrow \false
  \\
  \fire{test} &\Rightarrow \lefteqn{\fire{a}}
  & \fire{off_c} &\Rightarrow \false
  & \fire{test} &\Rightarrow \fire{a}
  & \fire{off_c} &\Rightarrow \fire{off}
  & \fire{test} &\Rightarrow \false
  & \fire{off_c} &\Rightarrow \false
\end{align*}

Applying the procedure from \cite{BliSif10-causal-fmsd}, we obtain the
$\ct(\act{P} \cup \fire{P} \cup \inhib{P})$ trees in \fig{example:trees}
and connectors in \fig{example:connectors}.  In terms of classical BIP, one
can easily distinguish here two priorities: $x\ a\ a_l \less b\ b_l$ and
$x\ off \less b\ b_l$ for all $x$ not containing $off\ off_c$.  In general,
priorities are replaced by local inhibitors.  In this example, these appear
to characterise states of the Main module.  For instance,
$\fire{a}\ \fire{a_l}\ \inhib{b}\ \inhib{off}$ defines possible
interactions involving $a\ a_l$ when neither $b$ nor $off$ are possible,
\ie in state 1 (see \fig{modules}).

\begin{figure*}
  \hfill
  \begin{subfigure}[b]{34mm}
    \begin{tikzpicture}
      \node (P) at (0,0) {$\fire{a}\ \fire{a_l}\ \inhib{b}\ \inhib{off}$};
      \node (Q) at (0,-2) {$\fire{test}$};
      \node[right] at (0.6, -1) {$\oplus\quad \fire{on}\ \fire{on_c}$};
      \path[->]
      (P) edge (Q);
    \end{tikzpicture}
  \end{subfigure}
  \hfill
  \begin{subfigure}[b]{53mm}
    \begin{tikzpicture}
      \node (P) at (0,0) {$\fire{off}\ \fire{off_c}$};
      \node (Q) at (0,-1) {$\fire{a}\ \fire{a_l}$};      
      \node (R) at (0,-2) {$\fire{test}$};
      \node at (1, -1) {$\oplus$};
      \node (T) at (2,0) {$\fire{a}\ \fire{a_l}\ \inhib{b}\ \inhib{on}$};
      \node (S) at (2,-2) {$\fire{test}$};
      \node[right] at (2.6, -1) {$\oplus\quad \fire{off}\ \inhib{b}$};
      \path[->]
      (P) edge (Q) 
      (Q) edge (R)
      (T) edge (S);      
    \end{tikzpicture}
  \end{subfigure}
  \hfill
  \begin{subfigure}[b]{2em}
    \begin{tikzpicture}
      \node at (0,-1) {$\fire{b}\ \fire{b_l}$};
      \node at (0,-2) {};
    \end{tikzpicture}
  \end{subfigure}
  \hfill\mbox{}
  \mycaption{Three causal interaction trees}
  \label{fig:example:trees}
\end{figure*}

\begin{figure*}
  \centering  
  \quad
  \begin{subfigure}[b]{46mm}
    \input{control-connector-1.pspdftex}
  \end{subfigure}
  \hfill
  \begin{subfigure}[b]{76mm}
    \input{control-connector-2.pspdftex}
  \end{subfigure}
  \hfill
  \begin{subfigure}[b]{8mm}
    \input{control-connector-3.pspdftex}
  \end{subfigure}
  \quad\mbox{}
  \mycaption{Connectors corresponding to trees from \fig{example:trees}}
  \label{fig:example:connectors}
\end{figure*}


\section{Conclusion}
\label{sec:conclusion}

The work presented in this paper relies on a variation of the BIP
operational semantics based on the offer predicate introduced in
\cite{BliSif11-constraints-sc}.  Glue operators defined using the offer
predicate are isomorphic to Boolean constraints on activation and firing
port variables $\sB[\act{P}, \fire{P}]$ with an additional axiom $\fire{p}
\Rightarrow \act{p}$ \cite{BliSif11-constraints-sc}.  By considering the
negation of an activation port variable as a separate {\em negative} port
variable (keeping, however, the axiom $\act{p}\,\inhib{p} = \false$), we
reinterpret the combination of interaction and priority models on $P$ as an
interaction model on $\act{P} \cup \fire{P} \cup \inhib{P}$.  This allows
us to apply the algebraic theory, previously developed for modelling
interactions in BIP, to model interactions and priorities simultaneously.
In particular, we can synthesise such new connectors from arbitrary
$\sB[\act{P}, \fire{P}]$ Boolean formul\ae\ (in \cite{BliSif10-causal-fmsd}
we have shown how to synthesise classical connectors from formul\ae\ on
port variables without firing/activation dichotomy).

The equivalence induced by the new operational semantics on the algebras
($x \sim y \bydef{\Leftrightarrow} \intsem{x}(\bB) = \intsem{y}(\bB)$ for
any finite set of behaviours $\bB$) is weaker than the standard equivalence
induced by the interaction semantics ($x \simeq y \bydef{\Leftrightarrow}
\aisem{x} = \aisem{y}$).  Extending the axioms of the Algebra of Causal
Interaction Trees accordingly, we define normal forms for connectors and
causal interaction trees.  This, in turn, allows us to simplify the causal
rule representation, by considering only rules with firing effects.
Algebra extensions are illustrated on a connector synthesis example.

In this paper, we have only extended the axiomatisation of $\ct(\act{P}
\cup \fire{P} \cup \inhib{P})$.  Studying corresponding extensions for the
axiomatisations of other algebras as well as their completeness could be
part of the future work.  More urgently, we intend to study the differences
between the classical BIP semantics and the offer variation.  For example,
it is clear that the two semantics are equivalent on flat models.  The
divergence on hierarchical models remains to be characterised.

\bibliographystyle{eptcs}
\bibliography{glue,express,constraints,connectors,reo,bip,latex8}

\end{document}